\documentclass[11pt,a4paper]{amsart}
\usepackage[english]{babel}
\usepackage{graphicx}

\newtheorem{theorem}{Theorem}[section]

\newtheorem{lemma}[theorem]{Lemma}
\newtheorem{proposition}[theorem]{Proposition}

\theoremstyle{definition}

\newtheorem{example}[theorem]{Example}
\newtheorem{remark}[theorem]{Remark}

\usepackage[square,sort,comma,numbers]{natbib}

\usepackage[all,cmtip,2cell,graph]{xy}
\usepackage{pstricks}
\usepackage{enumerate}
\usepackage{amsfonts,amssymb,amsmath,pinlabel,array,hhline}
\usepackage{slashed}
\usepackage{tabulary}
\usepackage{fancyhdr}
\usepackage{a4wide}
\usepackage{mathrsfs}
\usepackage{calrsfs}
\usepackage{bbm,dsfont}
\UseAllTwocells

\newcommand{\Z}{\mathds{Z}}

\newcommand{\R}{\mathds{R}}
\newcommand{\C}{\mathcal{C}}
\newcommand{\M}{\mathcal{M}}
\newcommand{\D}{\mathcal{D}}
\newcommand{\K}{\mathcal{K}}

\newcommand{\sign}{\mathrm{sign}}

\newcommand{\e}{\varepsilon}
\newcommand{\Pf}{\mathrm{Pf}}

\DeclareSymbolFont{EulerScript}{U}{eus}{m}{n}
\DeclareSymbolFontAlphabet\mathscr{EulerScript}
\begin{document}

\title{Identities between dimer partition functions on different surfaces}

\author{David Cimasoni}
\address{Universit\'e de Gen\`eve, Section de math\'ematiques, 2 rue du Li\`evre, 1211 Gen\`eve, Switzerland}
\email{david.cimasoni@unige.ch}
\author{Anh Minh Pham}
\address{Universit\'e de Gen\`eve, Section de math\'ematiques, 2 rue du Li\`evre, 1211 Gen\`eve, Switzerland}
\email{AnhMinh.Pham@unige.ch}


\begin{abstract}
Given a weighted graph~$G$ embedded in a non-orientable surface~$\Sigma$, one can consider the corresponding weighted graph~$\widetilde{G}$ embedded in the
so-called orientation cover~$\widetilde\Sigma$ of~$\Sigma$. We prove identities relating twisted partition functions of the dimer model on these two graphs.
When~$\Sigma$ is the M\"obius strip or the Klein bottle, then~$\widetilde\Sigma$ is the cylinder or the torus, respectively, and under some natural assumptions,
these identities imply relations between the genuine dimer partition functions~$Z(G)$ and~$Z(\widetilde{G})$. For example, we show that if~$G$ is a locally but not
globally bipartite graph embedded in the M\"obius strip, then~$Z(\widetilde{G})$ is equal to the square of~$Z(G)$. This extends results for the square lattice
previously obtained by various authors.
\end{abstract}

\maketitle


\section{Introduction}
\label{sec:intro}

\subsection{Background}

A {\em dimer configuration\/}, or {\em perfect matching\/}, on a finite graph~$G$ is a family of edges, called {\em dimers\/}, such that each vertex of~$G$ is
covered by exactly one of these dimers. To each edge~$e$ of~$G$, assign an {\em energy\/}~$\mathcal{E}(e)\in\R$ and for a dimer configuration~$D$,
write~$\mathcal{E}(D)=\sum_{e\in D}\mathcal{E}(e)$. The associated {\em dimer partition function\/} is then defined as
\[
Z(G)=\sum_{D\in\D(G)}\exp(-\mathcal{E}(D))\,,
\]
the sum being over the set~$\D(G)$ of dimer configurations on~$G$. The corresponding Boltzmann measure is the probability measure on~$\D(G)$ given
by~$\mu_G(D)=\frac{1}{Z(G)}\exp(-\mathcal{E}(D))$. The study of this measure, called the {\em dimer model\/}, is a very active field of research in statistical physics:
we refer the reader to~\cite{Ke2} for an introduction to this model, and to~\cite{Ke1} for a survey of recent results on the subject.
Note that the theory is empty unless the graph has an even number of vertices, which we will always assume.

While this model is defined on any abstract weighted graph~$G$, exact results are only available when this graph is embedded in the plane, or more
generally in a surface. In particular, computing the dimer partition function was shown to be~$\#P$-complete in general~\cite{Val}, but it can be achieved in
polynomial time if~$G$ is planar~\cite{Ka1}, or more generally, if it is embedded in a surface of fixed genus~\cite{Tes,G-L,C-R}.
To be more precise, it was first shown by Kasteleyn~\cite{Ka1} and Temperley-Fisher~\cite{T-F} that the dimer partition function for the~$m\times n$ square lattice~$G_{m,n}$
with free boundary conditions (i.e. embedded in the plane) is equal to the Pfaffian of the associated weighted skew-adjacency matrix defined with respect to a well-chosen orientation
of the edges of the lattice. Using properties of the Pfaffian, and assuming that all horizontal (resp. vertical) edges have the same weight~$x$ (resp.~$y$), these authors were
able to give a closed formula for~$Z(G_{m,n})=:Z^{\mathit{free}}_{m,n}(x,y)$. This result was then extended in two directions. On the one hand, Kasteleyn~\cite{Ka3} showed that
this {\em Pfaffian method\/}, i.e. the computation of the dimer partition function using Pfaffians,
is valid for any planar weighted graph. (This will be reviewed in Sections~\ref{sub:pf} and~\ref{sub:plane} below). On the other hand, the exact value for the partition function on the
square lattice (of width~$m$, length~$n$, and weights~$x,y$) was computed for various boundary conditions, or in other words, for embeddings in various surfaces. This was performed by Kasteleyn~\cite{Ka1} for periodic-periodic boundary conditions (which correspond to an embedding in the torus~$\mathbb{T}$), by McCoy-Wu~\cite{M-W}
for free-periodic ones (embedding in the cylinder~$\C$), by Brankov-Priezzhev~\cite{B-P} and Tesler~\cite{Tes} for free-antiperiodic ones (embedding in the M\"obius strip~$\M$)
and by Lu-Wu~\cite{L-W} for periodic-antiperiodic boundary conditions (embedding in the Klein bottle~$\K$). These boundary conditions (and surfaces) are best described using
a rectangle with identifications of pairs of sides encoded by arrows, as illustrated in Figure~\ref{fig:rectangles}. We shall use the
notations~$Z^{\mathbb{T}}_{m,n}(x,y)$,~$Z^\C_{m,n}(x,y)$,~$Z^{\M}_{m,n}(x,y)$ and~$Z^{\K}_{m,n}(x,y)$ for the corresponding partition functions. 

\begin{figure}[t]
\labellist\small\hair 2.5pt
\pinlabel {$\K$} at 930 230
\pinlabel {$\mathbb{T}$} at 90 230
\pinlabel {$\C$} at 355 230
\pinlabel {$\M$} at 640 230
\endlabellist
\centering
\includegraphics[width=0.9\textwidth]{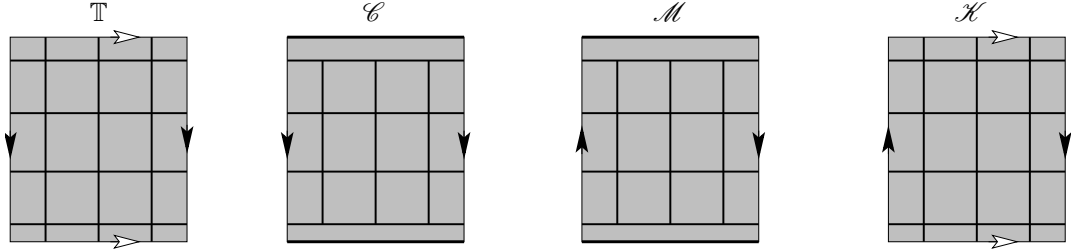}
\caption{The~$4\times 3$-square lattice embedded in the torus~$\mathbb{T}$, the cylinder~$\C$, the M\"obius strip~$\M$, and the Klein bottle~$\K$.}
\label{fig:rectangles}
\end{figure}

Using these results, Brankov-Priezzhev~\cite{B-P} proposed the identity
\begin{equation}
\label{equ:Mob1}
Z^\C_{2m,4n}(x,y)=\left(Z^\M_{2m,2n}(x,y)\right)^2\,,
\end{equation}
which they established in some large~$m,n$ expansion. This equality was later proved by Lu-Wu~\cite{L-W}, who also discovered the identity
\begin{equation}
\label{equ:Mob2}
Z^\C_{2m-1,4n}(x,y)=\frac{1}{2}\left(Z^\M_{2m-1,2n}(x,y)\right)^2\,.
\end{equation}
Finally Izmailian-Oganesyan-Hu~\cite{IOH} established the formula
\begin{equation}
\label{equ:K}
Z^\mathbb{T}_{2m-1,4n}(x,y)=\frac{1}{2}\left(Z^\K_{2m-1,2n}(x,y)\right)^2\,.
\end{equation}
Let us point out once again that these results were obtained by comparing explicit values of the various partition functions, values that are often available
only for the square lattice with weights~$x,y$. In particular, these authors regret the lack of ``general underlying principle'' explaining such ``curious identities''
(~\cite[p.650]{B-P} and~\cite[p.111]{L-W}, respectively).

In the present paper, we provide such a general principle that allows us, not only to explain these identities, but also to generalize them in quite a significant way.
We now summarize our results.

\subsection{Main results, and examples}

Given any non-orientable surface~$\Sigma$, there exists a~$2$-fold cover~$\widetilde{\Sigma}\to\Sigma$, called its {\em orientation cover\/}, which is determined by the
following property: a loop in~$\Sigma$ lifts to a loop in~$\widetilde{\Sigma}$ if and only if it admits a neighbourhood in~$\Sigma$ which is a cylinder. (Otherwise, it admits
a neighbourhood which is a M\"obius strip, and does not lift to a loop.) If~$\Sigma$ is the closed connected non-orientable surface of genus~$h\ge 1$ (the connected sum of~$h$
copies of the real projective plane), then its orientation cover is the orientable surface of genus~$g=h-1$ (the connected sum of~$g$ copies of the torus).
For example, removing a disc from the projective plane, one sees that the orientation cover of the M\"obius strip~$\M$ is the cylinder~$\C$ (see Figure~\ref{fig:coverM});
also, the torus~$\mathbb{T}$ is the orientation cover of the Klein bottle~$\K$ (see Figure~\ref{fig:coverK}). Obviously, any weighted graph~$G\subset\Sigma$ lifts to a weighted
graph~$\widetilde{G}\subset\widetilde{\Sigma}$, and the three identities displayed above all relate~$Z(G)$ to~$Z(\widetilde{G})$ for some specific examples of~$G\subset\Sigma$.

\begin{figure}[t]
\labellist\small\hair 2.5pt
\pinlabel {$=$} at 296 180
\pinlabel {$=$} at 1053 180
\pinlabel {$\to$} at 755 180
\pinlabel {$\C$} at 507 307
\pinlabel {$\M$} at 907 307
\endlabellist
\centering
\includegraphics[width=\textwidth]{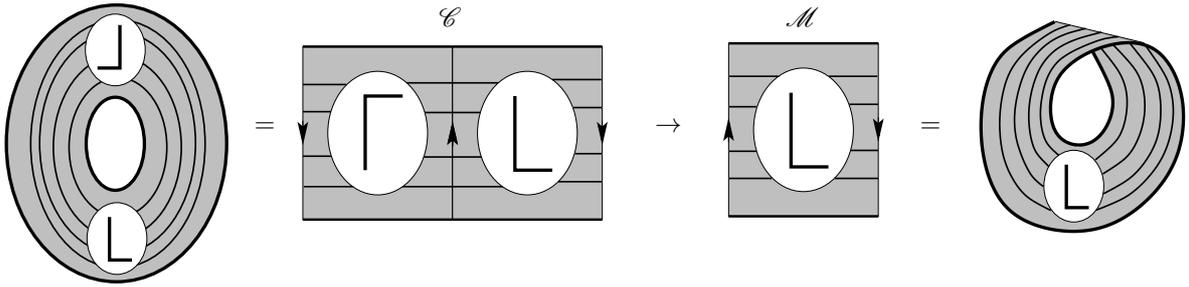}
\caption{A schematic description of the orientation cover of the M\"obius strip by the cylinder.}
\label{fig:coverM}
\end{figure}

Using the geometric approach to the Pfaffian method developped by the first-named author in~\cite{Cim}, we show that for any weighted graph~$G\subset\Sigma$ with an even number
of vertices, there are~$2^g=2^{h-1}$ distinct identities relating ``twisted versions'' of the dimer partition functions for~$G$ and~$\widetilde{G}$.
More precisely, these identities relate expressions of the form~$\sum_{D}\pm\exp(-\mathcal{E}(D))$, where the contribution of each dimer configuration~$D$ comes with a sign.
The precise definition of this sign, and therefore the precise statement of these general identities, use the terminology of {\em quadratic enhancements on surfaces\/}~\cite{K-T}.
Therefore, we shall not formulate them in a precise way here but refer the reader to Theorem~\ref{thm:main} below. However, let us briefly explain how they come about.
In a nutshell, there are essentially~$2^{h-1}$ orientations that can be used in the
Pfaffian method to study the dimer model on~$G$ embedded in a non-orientable surface~$\Sigma$ of genus~$h$. To any such orientation~$K$, one can associate an
orientation~$\widetilde{K}$ on the corresponding cover~$\widetilde{G}\subset\widetilde{\Sigma}$ which can be used to study the dimer model on~$\widetilde{G}$.
The two associated skew-adjacency matrices~$A^K(G)$ and~$A^{\widetilde{K}}(\widetilde{G})$ turn out to be related in a very simple way, which allows us to show that
their Pfaffians satisfy the equality~$|\Pf(A^{\widetilde{K}}(\widetilde{G}))|=|\Pf(A^{K}(G))|^2$ (see Lemma~\ref{lemma:Pf} below). As proved in~\cite{Cim}, these
Pfaffians are equal to twisted partition functions, so we obtain the~$2^{h-1}$ identities.

Let us consider the case of genus~$h=1$, i.e. of a graph~$G$ embedded in the M\"obius strip. In this case, we obtain a single identity, namely
\begin{equation}
\label{equ:Mob'}
Z(\widetilde{G})=Z_{0}(G)^2+Z_{1}(G)^2\,,
\end{equation}
where~$Z_{\alpha}(G)$ denotes the partial partition function given by~$Z_{\alpha}(G)=\sum_{[D\Delta D_0]=\alpha}\exp(-\mathcal{E}(D))$, the sum being over
all dimer configurations whose symmetric difference with a fixed~$D_0\in\D(G)$ winds around~$\M$ an even (resp. odd) number of times if~$\alpha=0$ (resp. if $\alpha=1$).
Of course, we are interested in identities between the actual partition functions. For this purpose, observe that Equation~(\ref{equ:Mob'})
implies that the identity~$Z(\widetilde{G})=Z(G)^2$ (resp.~$Z(\widetilde{G})=\frac{1}{2}Z(G)^2$) holds if and only if~$Z_{1}(G)=0$ (resp.~$Z_{0}(G)=Z_{1}(G)$). Therefore,
we are left with the task of finding natural conditions on~$G\subset\M$ for this to happen.

Recall that a graph embedded in a surface is {\em locally bipartite\/} (resp. {\em bipartite\/}) if the boundary of each face
(resp. if each cycle) has even length. The proof of the following theorem, which can be found in subsection~\ref{sub:proof1}, follows easily from the discussion above.

\begin{theorem}
\label{thm:1}
Let~$G\subset\M$ be a weighted graph with an even number of vertices embedded in the M\"obius strip, and let~$\widetilde{G}\subset\C$ denote its orientation cover.
\begin{enumerate}[(i)]
\item{If~$G$ is locally bipartite but not bipartite, then~$Z(\widetilde{G})=Z(G)^2$.}
\item{Assume that~$G\subset\M$ is invariant by a horizontal tranlation~$\tau$ of the M\"obius strip. If there is some~$D_0\in\D(G)$ such that
the symmetric difference~$D_0\Delta\tau(D_0)$ winds around~$\M$ an odd number of time, then~$Z(\widetilde{G})=\frac{1}{2}Z(G)^2$.}
\end{enumerate}
\end{theorem}

Examples are plentiful and include many classical regular lattices. In the context of this discussion, we shall only mention the following ones.

\begin{example}
\label{ex:Mob1}
Consider the case of the~$2m\times 2n$ square lattice~$G$ embedded in the M\"obius strip, endowed with arbitrary weights, and let~$\widetilde{G}$ be the
corresponding~$2m\times 4n$ square lattice in the cylinder. Since~$G$ is locally bipartite but not bipartite, we have the equality~$Z(\widetilde{G})=Z(G)^2$ which
extends Equation~(\ref{equ:Mob1}) to arbitrary weights.
\end{example}

\begin{example}
\label{ex:hexa}
Consider a portion of arbitrary size of the hexagonal lattice, embedded in the M\"obius strip as illustrated in the left-hand side of Figure~\ref{fig:ex},
and endowed with arbitrary weights. This graph~$G$ is easily seen not to be bipartite, so we have the equality~$Z(\widetilde{G})=Z(G)^2$.
\end{example}

\begin{example}
\label{ex:sqoct}
Let~$G$ be a portion of arbitrary size of the square-octogon lattice embedded in the M\"obius strip as illustrated in the center of Figure~\ref{fig:ex}.
This graph is locally bipartite but not bipartite, so the equality~$Z(\widetilde{G})=Z(G)^2$ holds for any weights on~$G$.
\end{example}

\begin{example}
\label{ex:Mob2}
Consider the~$(2m-1)\times 2n$ square lattice~$G$ embedded in the M\"obius strip, and endowed with any weight system that is invariant under horizontal translation by one edge.
Choosing for~$D_0$ any dimer configuration made of horizontal edges (this is possible since the lattice has even length), we find that~$D_0\Delta\tau(D_0)$
consists of all horizontal edges of~$G$, and hence winds around the M\"obius strip~$2m-1$ times. Therefore, we have the equality~$Z(\widetilde{G})=\frac{1}{2}Z(G)^2$ which
extends Equation~(\ref{equ:Mob2}).

This example can be generalized as follows: consider planar weighted graph~$G_1,\dots,G_{m-1}$ that admit a dimer configuration, and let~$\overline{G}_1,\dots,\overline{G}_{m-1}$
denote the weighted graphs obtained from these planar graphs via a horizontal reflexion. In each column of the square lattice, add a copy of~$G_1$ in the first face, a copy of~$G_2$
in the second, and so on until a copy of~$G_{m-1}$ in the~$(m-1)^\text{th}$ face, and then a copy of~$\overline{G}_{m-1}$, of~$\overline{G}_{m-2}$, until a copy of~$\overline{G}_1$
in the last face. Finally, join each of these graphs to the adjacent vertices of~$G$ in an arbitrary but fixed way for each~$G_i$. (The case~$m=2$ is illustrated in the right-hand side of Figure~\ref{fig:ex}.) By the second part of Theorem~\ref{thm:1}, the resulting graph~$G'$ will also satisfy the equality~$Z(\widetilde{G}')=\frac{1}{2}Z(G')^2$.
Note that this class of graphs contains the triangular lattice as a particularly natural example.

\end{example}

\begin{figure}[t]
\centering
\includegraphics[width=\textwidth]{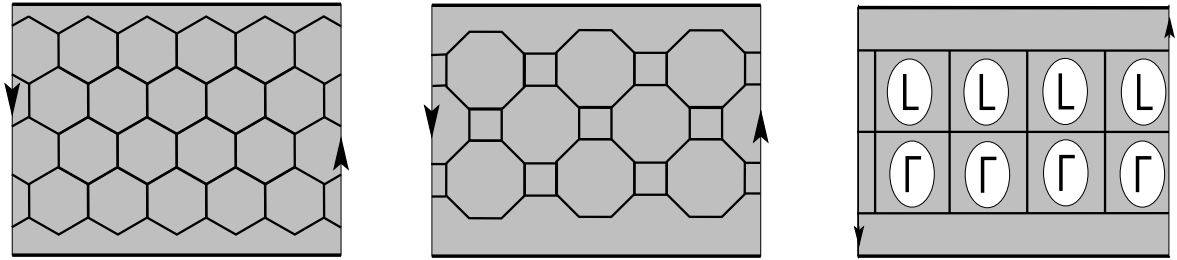}
\caption{The graphs in Examples~\ref{ex:hexa},~\ref{ex:sqoct} and~\ref{ex:Mob2}.}
\label{fig:ex}
\end{figure}

Let us now consider the case of genus~$h=2$, i.e. of a graph~$G$ embedded in the Klein bottle. In this case, we obtain two identities,
see Equation~(\ref{equ:h=2}) below. They easily lead to the following result, whose proof can be found in subsection~\ref{sub:general}.

\begin{theorem}
\label{thm:2}
Let~$G\subset\K$ be a weighted graph with an even number of vertices embedded in the Klein bottle, and let~$\widetilde{G}\subset\mathbb{T}$ denote its orientation cover.
Let us assume that~$G$ is locally bipartite but not bipartite, and invariant by a horizontal tranlation~$\tau$.
If there is some~$D_0\in\D(G)$ such that the symmetric difference~$D_0\Delta\tau(D_0)$ winds an odd number of times around~$\K$ in the horizontal direction (and an arbitrary
number of times in the vertical direction), then~$Z(\widetilde{G})=\frac{1}{2}Z(G)^2$.
\end{theorem}

\begin{example}
Consider the~$(2m-1)\times 2n$ square lattice embedded in the Klein bottle, and endowed with any weight system that is invariant under translation by one edge
along the horizontal direction in~$\K$. By the same arguments as in Examples~\ref{ex:Mob1} and~\ref{ex:Mob2}, one easily checks that it satisfies the hypothesis of Theorem~\ref{thm:2},
thus extending Equation~(\ref{equ:K}) to more general weights. (This can be further generalized as described at the end of Example~\ref{ex:Mob2}.)
\end{example}

\subsection{Consequences}

Let us begin with the combinatorial consequences of our results. Simply observe that if the energy of each edge is zero, then the dimer partition function~$Z(G)$ is equal
to the number of perfect matchings on~$G$. Therefore, each of the statements above implies an identity between the number of perfect matchings on~$G$ and
on~$\widetilde{G}$. For example, Theorem~\ref{thm:1} implies the following combinatorial statement:
{\em If~$G$ is a non-bipartite graph with an even number of vertices embedded in~$\M$ in a locally bipartite way, then the number of perfect matchings on its orientation
cover~$\widetilde{G}\subset\C$ is equal to the square of the number of perfect matchings on~$G$.\/}
To the best of our knowledge, such a statement cannot be proved by purely combinatorial means. (See~\cite{L-P} for the combinatorial theory of perfect matchings.)

We now come to the physical consequences of our results. Consider a two-dimensional system at criticality, with area~$A$ and boundary length~$L$. The total free
energy~$\log Z$ of such a system is commonly assumed~\cite{P-F} to have a large~$A$ expansion at fixed shape of the form
\[
F=Af_b + Lf_s +D+o(1)\,,
\]
where~$f_b$ is the bulk free energy (per unit surface area), which is independent of the boundary conditions,~$f_s$ is the surface free energy (per unit boundary length),
and~$D$ is a dimensionless coefficient. Furthermore, the coefficients~$f_b$ and~$f_s$ depend on the lattice, but~$D$ is expected to be universal, depending only on the shape
of the system, and possibly on the boundary conditions. (Analogies with conformal field theory lead several authors to conjecture the explicit form of this coefficient~$D$, see
e.g.~\cite{CMC} and references therein.)
As pointed out by Cardy-Peschel~\cite{C-P}, such an expansion is actually valid only for surfaces with vanishing
Euler characteristic, that is, precisely the four surfaces~$\M,\C,\K$ and~$\mathbb{T}$ considered above.

Applying the first part of Theorem~\ref{thm:1}, we find that such an expansion holds for a system on a locally bipartite non-bipartite graph~$G\subset\M$ if and only if it holds
for its orientation cover~$\widetilde{G}\subset\C$; in such a case both expansions have the same coefficients~$f_b$ and~$f_s$ (observe that the area and the boundary length
are doubled in~$\C$), while the corresponding universal coefficients are simply related by~$\widetilde{D}=2D$. (This was the precise result obtained by
Brankov-Priezzhev~\cite{B-P} for the square lattice with~$x,y$ weights.) Also, such an expansion holds for a graph~$G$ satisfying
the assumptions of the second part of Theorem~\ref{thm:1}, or of Theorem~\ref{thm:2}, if and only if it holds for its orientation cover~$\widetilde{G}$; in such a case both expansions have the same bulk free energy, as they should, the same surface free energy (which vanishes in the setting of Theorem~\ref{thm:2}), while the universal coefficients
are related by~$\widetilde{D}=2D-\log(2)$.

\subsection{Discussion, and organisation of the paper}

As is well-known since the work of Kasteleyn~\cite{Ka3}, the validity of the Pfaffian method for planar graphs can be demonstrated in a straightforward way.
However, the extension of this method to graphs embedded in non-planar surfaces requires either tedious combinatorial considerations~\cite{Tes,G-L}, or
additional geometric tools known as spin or pin$^-$ structures~\cite{C-R,Cim}.
It would have been possible to write a slightly shorter article relying heavily on~\cite{Cim} as a {\em black box\/} -- and the most general form of our identities actually still does.
However, in the course of our investigations, we discovered that the special geometry of simple closed curves
in the M\"obius strip allows an elementary proof of the Pfaffian method in this case as well. It thus became possible to give a simple self-contained proof of the identities
in the genus~$h=1$ case, leading in particular to an elementary demonstration of Theorem~\ref{thm:1}. Since such an elementary proof of the Pfaffian method for the M\"obius strip does not seem to
be known, we decided to include it in our paper, which is therefore organised in the following slightly unorthodox way: Section~\ref{sec:Mob} deals with the special case of the
M\"obius band, is entirely self-contained and meant to be readable without any knowledge of algebraic topology, while Section~\ref{sec:gen} deals with the general case,
relying heavily on previous work~\cite{Cim}.

To be more precise, we start in subsection~\ref{sub:pf} by recalling the well-known relation between dimers and
Pfaffians, leading to the notion of a Pfaffian orientation. In subsection~\ref{sub:plane}, we give a very short proof of Kasteleyn's theorem in the plane (Theorem~\ref{thm:Kast}),
and explain what changes for the cylinder. Subsection~\ref{sub:Mob} contains the aforementioned proof of the corresponding statement
in the M\"obius strip (Theorem~\ref{thm:Mob}), while subsection~\ref{sub:proof1} contains the proof of Equation~(\ref{equ:Mob'}), and of Theorem~\ref{thm:1}.
We then start Section~\ref{sec:gen} with a review of the notions of homology and of quadratic enhancements in subsection~\ref{sub:hom}. Finally, subsection~\ref{sub:general}
contains our main result (Theorem~\ref{thm:main}), its proof, and the proof of Theorem~\ref{thm:2}.

\subsection*{Acknowledgments}
This work was supported by a grant of the Swiss National Science Foundation.


\section{The identities for the cylinder and M\"obius strip}
\label{sec:Mob}

The aim of this section is to give a self-contained treatement of the Pfaffian method (subsection~\ref{sub:pf}) for graphs in the cylinder (subsection~\ref{sub:plane}) and in the
M\"obius strip (subsection~\ref{sub:Mob}), leading to the proof of Theorem~\ref{thm:1} (subsection~\ref{sub:proof1}).

Throughout this section,~$G$ is a finite graph with an even number of vertices, and each edge~$e$ of~$G$ is endowed with a real number~$\mathcal{E}(e)$. To avoid
cumbersome notation, we shall write~$\nu(e)=\exp(-\mathcal{E}(e))$, and~$\nu(D)=\prod_{e\in D}\nu(e)$ for any family~$D$ of edges of~$G$. Note that the weights~$\nu(e)$
can be considered as formal variables, as all of our identities hold true in the ring~$\Z[i][\{\nu(e)\}_e]$.

\subsection{Dimers and Pfaffians}
\label{sub:pf}
We begin by explaining how Pfaffians are related to dimer partition functions, as first discovered by Kasteleyn~\cite{Ka1} and Temperley-Fisher~\cite{T-F}.

Recall that the determinant of a skew-symmetric matrix~$A=(a_{ij})$ of size~$2n$ is the square of a polynomial in the~$a_{ij}$'s.
This square root, called the {\em Pfaffian\/} of~$A$, is given by
\[
\Pf(A)=\frac{1}{2^n n!}\sum_{\sigma\in S_{2n}}\sign(\sigma) a_{\sigma(1)\sigma(2)}\cdots a_{\sigma(2n-1)\sigma(2n)},
\]
where the sum is over all permutations of~$\{1,\dots,2n\}$ and~$\sign(\sigma)=\pm 1$ denotes the signature of~$\sigma$. Since~$A$ is skew-symmetric, the monomial
corresponding to~$\sigma$ only depends on the matching of~$\{1,\dots,2n\}$ into unordered pairs~$\{\sigma(1),\sigma(2)\},\dots,\{\sigma(2n-1),\sigma(2n)\}$.
As there are~$2^n n!$ different permutations defining the same matching, we get
\[
\Pf(A)=\sum_{[\sigma]}\sign(\sigma) a_{\sigma(1)\sigma(2)}\cdots a_{\sigma(2n-1)\sigma(2n)},
\]
where the sum is on the set of matchings of~$\{1,\dots,2n\}$. Note that there exists a skew-symmetric version of the Gauss elimination algorithm, which
allows us to compute the Pfaffian of a matrix of size~$2n$ in~$O(n^3)$ time.

Pfaffians can be used to compute the dimer partition function of a weighted graph~$(G,\nu)$, as follows. Order totally the~$2n$ vertices of~$G$ and fix an arbitrary
orientation~$K$ of the edges of~$G$. Let~$A^K(G)=(a_{uv})$ denote the associated weighted skew-adjacency matrix, i.e. the skew-symmetric matrix whose rows and columns
are indexed by the vertices of~$G$, and whose coefficients are given by
\begin{equation}
\label{equ:A}
a_{uv}=\sum_{e=(u,v)}\e_{uv}^K(e)\nu(e),
\end{equation}
where the sum is over all edges~$e$ in~$G$ between the vertices~$u$ and~$v$, and
\[
\e^K_{uv}(e)=
\begin{cases}
\phantom{-}1 & \text{if~$e$ is oriented by~$K$ from~$u$ to~$v$;} \\
-1 & \text{otherwise}\,.
\end{cases}
\]
Now, let us consider the Pfaffian of this matrix. A matching of the vertices of~$G$ contributes to~$\Pf(A^K(G))$ if and only if it is realized by a dimer configuration
on~$G$, and this contribution is~$\pm\nu(D)$. More precisely,
\begin{equation}
\label{equ:Pf}
\Pf(A^K(G))=\sum_{D\in\D(G)}\e^K(D)\nu(D)\,,
\end{equation}
where the sign~$\e^K(D)$ can be computed as follows: if the dimer configuration~$D$ is given by edges~$e_1,\dots,e_n$ matching vertices~$u_\ell$ and~$v_\ell$ for~$\ell=1,\dots,n$,
let~$\sigma$ denote the permutation mapping the totally ordered set of vertices of~$G$ to~$(u_1,v_1,\dots,u_n,v_n)$; the sign is equal to
\[
\e^K(D)=\sign(\sigma)\prod_{\ell=1}^n\e^K_{u_\ell v_\ell}(e_\ell)\,.
\]
The problem of expressing~$Z(G)$ as a Pfaffian now boils down to finding an orientation~$K$ of the edges of~$G$ such that~$\e^K(D)$ does not depend on~$D$.

Let us therefore fix~$D,D'\in\D(G)$ and try to compute the product~$\e^K(D)\e^K(D')$. Note that the symmetric difference~$D\Delta D'$ consists of a disjoint union of simple
closed curves~$C_1,\dots,C_m$ of even length in~$G$. Since the matchings~$D$ and~$D'$ alternate along these cycles, one can choose permutations~$\sigma$ (resp.~$\sigma'$)
representing~$D$ (resp.~$D'$) such that~$s=\sigma'\circ\sigma^{-1}$ is the product of the cyclic permutations defined by the cycles~$C_1,\dots,C_m$. Using this particular choice of representatives, and the fact that~$\sign(s)=(-1)^m$, we find
\begin{equation}
\label{equ:cc}
\e^K(D)\e^K(D')=\prod_{j=1}^m(-1)^{n^K(C_j)+1},
\end{equation}
where~$n^K(C_j)$ denotes the number of edges of~$C_j$ where a fixed orientation of~$C_j$ differs from~$K$.
(Since~$C_j$ has even length, the parity of this number is independent of the chosen orientation of~$C_j$.)

Fixing a reference matching~$D_0$, Equations~(\ref{equ:Pf}) and~(\ref{equ:cc}) lead to the equality
\begin{equation}
\label{equ:Pf'}
\e^K(D_0)\,\Pf(A^K(G))=\sum_{D\in\D(G)}(-1)^{\sum_i(n^K(C_i)+1)}\,\nu(D)\,,
\end{equation}
where~$D\Delta D_0=\bigsqcup_j C_j$. Therefore, we are now left with the problem of finding an orientation~$K$ of~$G$ with the following property: for any cycle~$C$ of even
length such that~$G\setminus C$ admits a dimer configuration,~$n^K(C)$ is odd. By Equation~(\ref{equ:Pf'}), if~$K$ is such a {\em Pfaffian orientation\/},
then~$Z(G)=|\Pf(A^K(G))|$, which can be computed in polynomial time.

\subsection{Graphs in the plane and in the cylinder}
\label{sub:plane}

Kasteleyn's celebrated theorem asserts that any planar graph admits a Pfaffian orientation, and we include a very short proof for the sake of completeness.
More precisely, let~$G$ be a graph embedded in the plane. We shall say that an orientation~$K$ on~$G$ is a {\em Kasteleyn orientation\/} on~$G\subset\R^2$ if,
for each face~$f$ of~$G\subset\R^2$,~$n^K(\partial f)$
is odd, where~$\partial f$ denotes the boundary of~$f$ oriented counterclockwise.

\begin{theorem}[Kasteleyn~\cite{Ka3}]
\label{thm:Kast}
Any planar graph admits a Kasteleyn orientation~$K$, and any such orientation is Pfaffian, so~$|\Pf(A^K(G))|=Z(G)$.
\end{theorem}

\begin{proof}
Fix a graph~$G$ embedded in the plane, together with an arbitrary orientation~$K$ of its edges. If a face~$f$ of~$G\subset\R^2$ is such that~$n^{K}(\partial f)$ is even,
draw a path from the interior of~$f$ to the outer face, transverse to~$G$, and invert~$K$ on each edge crossed by this path.
Repeating this procedure for each face with~$n^{K}(\partial f)$ even, we obtain a Kasteleyn orientation.
To prove that such a Kasteleyn orientation~$K$ on~$G$ is Pfaffian, let us fix a cycle~$C\subset G$ of even length such that~$G\setminus C$ admits a dimer configuration.
Since~$C$ is a simple closed plane curve, it bounds a closed disc, with~$V$ vertices,~$E$ edges, and~$F$ faces. Let us write~$V=V_\text{int}+V_\text{ext}$
and~$E=E_\text{int}+E_\text{ext}$, where~$V_\text{int}$ (resp.~$E_\text{int}$) denotes the number of vertices (resp. edges) in the interior of the disc. Note that since~$C$ is
simple and closed,~$V_\text{ext}$ and~$E_\text{ext}$ coincide with the length of~$C$, and are therefore even. Also, since~$G\setminus C$ admits a dimer configuration,~$V_\text{int}$
is even as these interior vertices are matched by this dimer configuration. Summing over all the faces of this disc, and computing modulo~$2$, we therefore have
\[
0=\sum_f(n^K(\partial f)+1)=n^K(C)+E_\text{int}+F=n^K(C)+V+E+F=n^K(C)+1\,,
\]
since the Euler characteristic of the disc is~$V-E+F=1$. This completes the proof.
\end{proof}

Let us turn to the case of a graph~$G$ embedded in a cylinder~$\mathcal{C}$. Very naturally, we shall say that an orientation~$K$ is a {\em Kasteleyn orientation\/}
on~$G\subset\C$ if~$n^K(\partial f)$ is odd for each face~$f$ of~$G\subset\C$, where~$\partial f$ is oriented via a fixed orientation of the cylinder. Note however that in this case,
a Kasteleyn orientation is not always Pfaffian. Indeed, embedding~$\C$ in the plane,~$\R^2\setminus\C$ has a bounded component~$f_0$, and~$n^K(\partial f_0)$ might be odd, or even.
In the former case,~$K$ is a Kasteleyn orientation for~$G\subset\R^2$, and therefore a Pfaffian orientation by Theorem~\ref{thm:Kast}. In the later case, however,
the Pfaffian of the corresponding matrix~$A^K(G)$ counts some signed partition function that does not equal~$Z(G)$ in general. Luckily for us, the Kasteleyn orientations that will
naturally appear in our computations will always turn out to be of the former type.

\subsection{Graphs in the M\"obius strip}
\label{sub:Mob}

We now turn to the case of a graph~$G$ embedded in a M\"obius strip~$\M$. This situation is more tricky and involves the introduction of complex valued coefficients in
the skew-adjacency matrices, as first noticed by Tesler~\cite{Tes}. Here, we give a new elementary proof of this result, adapting the geometrical treatment given by the first-named
author in~\cite{Cim}.

In what follows, we will always represent the M\"obius strip as a rectangle with antiperiodic boundary conditions on the vertical sides, and free boundary conditions on the
horizontal sides (recall Figure~\ref{fig:rectangles}). The graph~$G$ is embedded in this rectangle with some specific edges intersecting transversally the vertical sides,
and no edge intersecting the horizontal sides. We assign to each such specific edge~$e$ the value~$\omega(e)=1$, and to all the others the
value~$\omega(e)=0$. As usual, we also write~$\omega(D)=\sum_{e\in D}\omega(e)$, and similarly for any set of edges.

Given an arbitrary orientation~$K$ on the edges of~$G\subset\M$, let~$A^{K,\omega}=(a_{uv})$ denote the modified weighted skew-adjacency matrix given by
\begin{equation}
\label{equ:A'}
a_{uv}=\sum_{e=(u,v)}\e_{uv}^K(e)\,i^{\omega(e)}\nu(e)\,.
\end{equation}
Following the notation and the discussion of subsection~\ref{sub:pf} up to Equation~(\ref{equ:Pf'}), we easily find
\begin{equation}
\label{equ:Pf''}
i^{-\omega(D_0)}\,\e^K(D_0)\,\Pf(A^{K,\omega}(G))=\sum_{D\in\D(G)}(-1)^{\sum_j(n^K(C_j)+1)}\,i^{\sum_j\omega(C_j\setminus D_0)-\omega(C_j\cap D_0)}\,\nu(D)\,,
\end{equation}
where~$D\Delta D_0=\bigsqcup_j C_j$.

The next step is to define the right notion of a Kasteleyn orientation is this non-orientable setting. This can be done as follows.
Given a graph~$G\subset\M$ as described above, let~$\widetilde{G}\subset\C$ denote the graph embedded in the cylinder obtained by taking two copies of~$G\subset\M$
and gluing them along the vertical sides, as illustrated in Figure~\ref{fig:coverM}.
Also, if~$K$ is an orientation on~$G$, let~$\widetilde{K}$ denote the orientation
on~$\widetilde{G}$ obtained by lifting~$K$ to the edges of~$\widetilde{G}$, and by inverting the orientation of all the edges that are completely contained in the second
copy of the M\"obius strip. We shall say that~$K$ is a {\em Kasteleyn orientation\/} on~$G\subset\M$ if~$\widetilde{K}$ is a Kasteleyn orientation on~$\widetilde{G}\subset\C$,
as defined in subsection~\ref{sub:plane}. A didactic example is given in Figure~\ref{fig:K}.
Note that each face~$f$ of~$G\subset\M$ lifts to two faces~$\widetilde{f}'$ and~$\widetilde{f}''$ of~$\widetilde{G}\subset\C$, and by definition,~$\widetilde{K}$
satisfies the parity condition around~$\widetilde{f}'$ if and only if it does around~$\widetilde{f}''$. Then, one can check that any~$G\subset\M$ admits a Kasteleyn orientation,
as in the planar case.
Note that this notion not only depends on the way~$G$ is embedded in~$\M$, but also on the way the M\"obius strip is
drawn as a rectangle, which is encoded by~$\omega$. Therefore, we shall sometimes be more precise and speak of a Kasteleyn orientation on~$(G\subset\M,\omega)$.

\begin{figure}[t]
\labellist\small\hair 2.5pt
\pinlabel {$\C$} at -20 105
\pinlabel {$\M$} at 630 105
\pinlabel {$\to$} at 388 105
\pinlabel {$K$} at 517 180
\pinlabel {$\widetilde K$} at 169 180
\endlabellist
\centering
\includegraphics[width=0.5\textwidth]{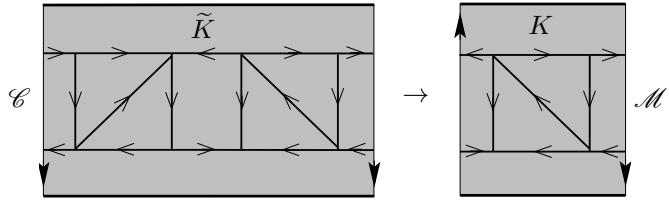}
\caption{An example of a Kasteleyn orientation~$K$ on some graph~$G$ embedded in the M\"obius strip, and of the corresponding Kasteleyn orientation~$\widetilde K$ on the
orientation cover~$\widetilde{G}\subset\C$.}
\label{fig:K}
\end{figure}

Let us define~$Z_{\alpha,D_0}(G)=\sum_{[D\Delta D_0]=\alpha}\nu(D)$, where the sum is over all~$D\in\D(G)$
such that the parity of the number of times~$D\Delta D_0$ winds around the M\"obius strip is given by~$\alpha\in\Z_2$.
We are now ready to state the main result of this paragraph.

\begin{theorem}
\label{thm:Mob}
Fix an arbitrary dimer configuration~$D_0$ on a graph~$G$ embedded in a M\"obius strip~$\M$. Then, for any Kasteleyn orientation~$K$ on~$(G\subset\M,\omega)$,
\[
i^{-\omega(D_0)}\,\e^K(D_0)\,\Pf(A^{K,\omega}(G))=Z_{0,D_0}(G)\pm i\,Z_{1,D_0}(G)\,.
\]
\end{theorem}

\begin{proof}
We will first show that this equality does not depend on the way the M\"obius strip is cut open into a rectangle, i.e. on the choice of~$\omega$. (This is evident for the right-hand side.) More precisely, any two choices of~$\omega$ are related by a sequence of ``vertex flips'', consisting in changing the value of~$\omega$ at each edge adjacent to a fixed
vertex~$v$. (Geometrically, this corresponds to moving the cut across the vertex~$v$.) We claim that if~$\omega'$ is obtained from~$\omega$ by such a vertex flip,
then there is a canonical way to transform a Kasteleyn orientation~$K$ on~$(G\subset\M,\omega)$ into a Kasteleyn orientation~$K'$ on~$(G\subset\M,\omega')$ so
that~$i^{-\omega(D_0)}\,\e^K(D_0)\,\Pf(A^{K,\omega}(G))$ is left unchanged. Indeed, define~$K'$ as coinciding with~$K$ everywhere, except
precisely on the edges~$e$ adjacent to~$v$ such that~$\omega(e)=1$. One easily checks that~$K'$ is indeed a Kasteleyn orientation on~$(G\subset\M,\omega')$,
that~$\Pf(A^{K',\omega'}(G))=i\,\Pf(A^{K,\omega}(G))$ and that~$i^{-\omega'(D_0)}\,\e^{K'}(D_0)=(-i)\,i^{-\omega(D_0)}\,\e^{K}(D_0)$. This proves the claim.

As a second step, let us recall a couple of well-known geometric facts. Simple closed curves in the M\"obius strip~$\M$ fall into three categories: they either bound a disc,
or bound a M\"obius strip isotopic to~$\M$, or do not bound, in which case they wind around~$\M$ exactly once. Furthermore, any two curves that do not bound must intersect
each other. Therefore, given any two dimer configurations~$D,D_0$ on~$G\subset\M$, their symmetric difference~$D\Delta D_0=\bigsqcup_j C_j$ consists of a certain number of
simple closed curves of the first two types, and at most one of the third type.

By Equation~(\ref{equ:Pf''}), we therefore need to show the following claim: if~$D$ is such that~$D\Delta D_0$ bounds discs and M\"obius strips isotopic to~$\M$, then
\[
(-1)^{\sum_j(n^K(C_j)+1)}\,i^{\sum_j\omega(C_j\setminus D_0)-\omega(C_j\cap D_0)}=1\,.
\]
Let us fix~$D$ as above. Clearly, one can cut~$\M$ open into a rectangle without touching any of these discs, without intersecting any of the edges of~$D_0$, and cutting each
of the M\"obius strips open into a rectangle. Choosing the~$\omega$ which corresponds to such a cut, as permitted by the first part of the proof, the equality displayed above is now a consequence of the following two statements:
\begin{enumerate}[(i)]
\item{if~$C\subset D\Delta D_0$ bounds a disc, then~$n^K(C)$ is odd (note that~$\omega(C)=0$ in this case);}
\item{if~$C\subset D\Delta D_0$ bounds a M\"obius strip, then~$n^K(C)$ is even (note that here,~$\omega(C)=2$).}
\end{enumerate}
The proof of the first statement follows {\em verbatim\/} the proof of Kasteleyn's theorem given above (Theorem~\ref{thm:Kast}). To show the second statement, let us denote
by~$V$,~$E$, and~$F$ the number of vertices, edges and faces, respectively, of the M\"obius strip~$\M'$ bounded by~$C$, and let us write~$V=V_\text{int}+V_\text{ext}$
and~$E=E_\text{int}+E_\text{ext}$, where~$V_\text{int}$~(resp.~$E_\text{int}$) denotes the number of vertices (resp. edges) in the interior of~$\M'$.
The cutting open of~$\M$ into a rectangle defines a decomposition of the oriented simple closed curve~$C$ bounding~$\M'$ into two oriented simple curves~$C'$ and~$C''$ in the rectangle, as illustrated in Figure~\ref{fig:M}. Summing over all the faces of~$\M'$,
and using the definition of a Kasteleyn orientation in a M\"obius strip, we get
\[
\sum_f(n^K(\partial f)+1)=\sum_f(n^{\widetilde{K}}(\partial \widetilde{f})+1)+\ell=\ell\,,
\]
with~$\ell$ denoting the number of edges bounding faces of~$\M$ where the orientation~$K$ lifted to the cylinder is inverted to give~$\widetilde{K}$.
(In Figure~\ref{fig:M}, there are~$\ell=3$ such edges, that are drawn with heavier lines.) On the other hand, computing modulo~$2$,
this same sum is equal to
\[
\sum_f(n^K(\partial f)+1)=n^K(C')+n^K(-C'')+\ell+E_\text{int}+F\,.
\]
These two equations lead to the equality modulo~$2$
\begin{equation}
\label{equ:Mob}
0=n^K(C')+n^K(-C'')+E_\text{int}+F=n^K(C)+|C''|+E_\text{int}+F=n^K(C)+|C'|+V\,,
\end{equation}
using the facts that~$|C'|+|C''|=E_\text{ext}$ and that the Euler characteristic of the M\"obius strip is equal to~$V-E+F=0$.
Note that since~$D_0$ does not meet the cut, the length of both~$C'$ and~$C''$ is even, and so is~$V_\text{ext}$. Finally,~$V_\text{int}$ is also even since these interior
vertices are matched by~$D_0$. By Equation~(\ref{equ:Mob}),~$n^K(C)$ is therefore even, and the claim is proved.

\begin{figure}[t]
\labellist\small\hair 2.5pt
\pinlabel {$C'$} at 200 70
\pinlabel {$C''$} at 230 160
\endlabellist
\centering
\includegraphics[width=0.4\textwidth]{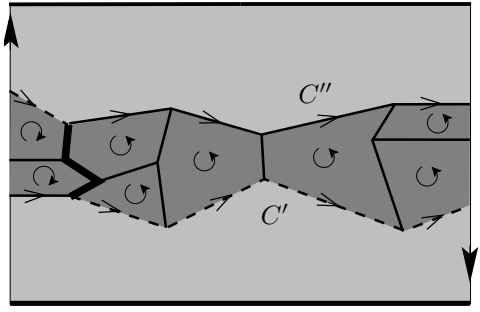}
\caption{Proof of Theorem~\ref{thm:Mob}: the shaded M\"obius strip~$\M'$ is bounded by the oriented simple curves~$C'$ in dashed lines and~$C''$ in solid lines, while
the~$\ell=3$ special edges are in heavier solid lines.}
\label{fig:M}
\end{figure}

In summary, we showed that for any~$D_0\in\D(G)$ and for any Kasteleyn orientation on~$(G\subset\M,\omega)$, we have the equality
\[
\textstyle i^{-\omega(D_0)}\,\e^K(D_0)\,\Pf(A^{K,\omega}(G))=\sum_{[D\Delta D_0]=0}\nu(D)+i\sum_{[D\Delta D_0]=1}\pm\nu(D)\,.
\]
If there is no~$D$ such that~$[D\Delta D_0]=1$, then the proof is complete. Otherwise, fix such a dimer configuration~$D_1$. Note that~$\omega(D_0)$ and~$\omega(D_1)$ have different
parity, and that for any~$D\in\D(G)$,~$[D\Delta D_0]$ and~$[D\Delta D_1]$ are different. Therefore, applying the equality displayed above to~$D_1$, we get
\begin{align*}
\textstyle i^{-\omega(D_0)}\,\e^K(D_0)\,\Pf(A^{K,\omega}(G))&=\pm i\cdot i^{-\omega(D_1)}\,\e^K(D_1)\,\Pf(A^{K,\omega}(G))\\
	&=\textstyle\pm i\left(\sum_{[D\Delta D_1]=0}\nu(D)+i\sum_{[D\Delta D_1]=1}\pm\nu(D)\right)\\
	&=\textstyle\sum_{[D\Delta D_0]=0}\pm\nu(D)\pm i\sum_{[D\Delta D_0]=1}\nu(D)\,.
\end{align*}
The two equations displayed above yield the statement of the theorem.
\end{proof}

As an immediate corollary, we get that the dimer partition function of a graph embedded in a M\"obius strip can be computed using a single Pfaffian.
However, we are interested in another consequence, namely Theorem~\ref{thm:1}.

\subsection{Proof of the identities for the cylinder and M\"obius strip}
\label{sub:proof1}

Our proof relies on the following easy but crucial lemma.

\begin{lemma}
\label{lemma:Pf}
Let~$K$ be any Kasteleyn orientation on~$G\subset\M$, and~$\widetilde{K}$ the corresponding orientation on~$\widetilde{G}\subset\C$.
If~$G$ has an even number of vertices, then~$|\Pf(A^{\widetilde{K}}(\widetilde{G}))|=|\Pf(A^{K,\omega}(G))|^2$.
\end{lemma}

\begin{proof}
Enumerate the vertices of~$\widetilde{G}$ contained in one copy of~$G$, followed by the vertices contained in the other copy with the same order. Recall that~$\widetilde{K}$ is
obtained by lifting~$K$ to the edges of~$\widetilde{G}$ and by inverting the orientation of all the edges that are completely contained in the second copy of the M\"obius strip.
Using the definitions~$(\ref{equ:A})$ and~$(\ref{equ:A'})$ of the corresponding skew-adjacency matrices, we obtain
\[
A^{\widetilde{K}}(\widetilde{G})=\begin{pmatrix}M_1&\phantom{-}M_2\cr M_2& -M_1\end{pmatrix}\quad\text{and}\quad A^{K,\omega}(G)=M_1+i M_2\,,
\]
with~$M_1,M_2$ real square matrices of even dimension. Using obvious operations, we get
\[
\det A^{\widetilde{K}}(\widetilde{G})=\begin{vmatrix}M_1&iM_2\cr iM_2&M_1\end{vmatrix}=\begin{vmatrix}M_1+iM_2&iM_2\cr M_1+iM_2&M_1\end{vmatrix}
=\begin{vmatrix}M_1+iM_2&iM_2\cr 0&M_1-iM_2\end{vmatrix}=\left|\det A^{K,\omega}(G)\right|^2\,.
\]
Since the determinant is the square of the Pfaffian, this proves the lemma.
\end{proof}

\begin{proposition}
\label{prop:Mob}
Let~$G\subset\M$ be a weighted graph embedded in the M\"obius strip, and let~$\widetilde{G}\subset\C$ denote its 2-fold cover embedded in the cylinder.
If~$G$ has an even number of vertices, then
\[
Z(\widetilde{G})=Z_{0,D_0}(G)^2+Z_{1,D_0}(G)^2
\]
for any dimer configuration~$D_0$ on~$G$.
\end{proposition}

\begin{proof}
Using Theorem~\ref{thm:Mob} and Lemma~\ref{lemma:Pf}, we have
\[
Z_{0,D_0}(G)^2+Z_{1,D_0}(G)^2=\left|\Pf(A^{K,\omega}(G))\right|^2=|\Pf(A^{\widetilde{K}}(\widetilde{G}))|\,.
\]
By Theorem~\ref{thm:Kast}, we only need to check that, given any Kasteleyn orientation~$K$ on~$G\subset\M$, the corresponding Kasteleyn orientation~$\widetilde{K}$
on~$\widetilde{G}\subset\C$ is also a Kasteleyn orientation on~$\widetilde{G}\subset\R^2$. In other words, writing~$f_0$ for the bounded component of~$\C\subset\R^2$,
we need to verify that~$n^{\widetilde{K}}(\partial f_0)$ is odd. Let~$C$ denote the simple closed curve in~$G$ bounding the M\"obius strip~$\M$. Cutting~$\M$ open into
a rectangle decomposes~$C$ into two simple curves~$C'$ and~$C''$ (recall Figure~\ref{fig:M}), and by definition of~$\widetilde{K}$, we have the
equality~$n^{\widetilde{K}}(\partial f_0)=n^K(C)+|C'|+1$ in~$\Z_2$. (Note that one of the lifts of~$C$ is equal to~$\partial f_0$.)
Therefore, we need to check that~$n^K(C)$ and~$|C'|$ have the same parity. By Equation~(\ref{equ:Mob}), this is the case if and only if~$G$ contains an even number of vertices,
which we assumed. 
\end{proof}

The proof of Theorem~\ref{thm:1} is now straightforward.

\begin{proof}[Proof of Theorem~\ref{thm:1}]
If~$G\subset\M$ is locally bipartite but not globally bipartite, then any simple closed curve in~$G$ that does not bound in~$\M$ is of odd length.
Hence, there is no dimer configuration~$D$ such that~$[D\Delta D_0]=1$, and~$Z_{1,D_0}(G)$ vanishes.
By Proposition~\ref{prop:Mob},~$Z(\widetilde{G})=Z_{0,D_0}(G)^2=Z(G)^2$, proving the first point.

To show the second point, let us assume that~$G\subset\M$ in invariant by a horizontal translation~$\tau$, and that for some~$D_0\in\D(G)$,~$[D_0\Delta\tau(D_0)]=1$.
Then,~$\tau^{-1}$ defines a bijection of~$\D(G)$ onto itself with~$\nu(\tau^{-1}(D))=\nu(D)$ for all~$D\in\D(G)$.
Furthermore, we have the equalities in~$\Z_2$:
\[
[\tau^{-1}(D)\Delta D_0]=[D\Delta\tau(D_0)]=[D\Delta D_0]+[D_0\Delta\tau(D_0)]=[D\Delta D_0]+1\,.
\]
In other words,~$\tau^{-1}$ maps the dimer configurations contributing to~$Z_{0,D_0}(G)$ to the ones contributing to~$Z_{1,D_0}(G)$, and vice-versa. This implies an
equality between these two quantities. The result now follows from Proposition~\ref{prop:Mob} together with the trivial equality~$Z(G)=Z_{0,D_0}(G)+Z_{1,D_0}(G)$.
\end{proof}


\section{The identities for surfaces of arbitrary genus}
\label{sec:gen}

The aim of this section is to prove our main result in its most general form, using the results of~\cite{Cim}. We start in subsection~\ref{sub:hom} by recalling the
necessary terminology, while subsection~\ref{sub:general} contains the main theorem and its proof, as well as the proof of Theorem~\ref{thm:2}.

\subsection{Homology, orientation covers, and quadratic enhancements}
\label{sub:hom}

In all this paragraph,~$\Sigma$ denotes a closed connected surface, not necessarily orientable.

Given a graph~$G\subset\Sigma$ whose complement consists of topological discs, let~$C_0$ (resp.~$C_1$,~$C_2$) denote the~$\Z_2$-vector space with basis the set of vertices
(resp. edges, faces) of~$G\subset\Sigma$. Also, let~$\partial_2\colon C_2\to C_1$ and~$\partial_1\colon C_1\to C_0$ denote the boundary operators defined in the obvious way.
Since~$\partial_1\circ\partial_2$ vanishes, the space of cycles~$\mathrm{ker}(\partial_1)$ contains the space~$\partial_2(C_2)$ of boundaries.
Therefore, one can define the {\em first homology space\/}~$H_1(\Sigma;\Z_2)$ as the quotient~$\mathrm{ker}(\partial_1)/\partial_2(C_2)$. This space turns out not to depend
on~$G$, but only on~$\Sigma$ (see~\cite{Hat} for details). Also, the intersection of curves defines a non-degenerate bilinear form on~$H_1(\Sigma;\Z_2)$, that will be denoted
by~$(\alpha,\beta)\mapsto \alpha\cdot\beta$. Finally, note that any simple closed curve in~$\Sigma$ admits a bicollar neighborhood
which is either a cylinder or a M\"obius strip. Assigning to this cycle the value~$0$ in the first case and~$1$ in the second gives a well-defined map in homology, that is usually
denoted by~$w_1\colon H_1(\Sigma;\Z_2)\to\Z_2$.

Let us be more specific. If~$\Sigma$ is orientable, then there is a unique integer~$g\ge 0$ such that~$\Sigma$ is homeomorphic to the connected sum of~$g$ tori (the sphere if~$g=0$).
In such a case, the space~$H_1(\Sigma;\Z_2)$ has dimension~$2g$, the intersection form is given by the
matrix~$\left(\begin{smallmatrix}0&1\cr 1&0\end{smallmatrix}\right)^{\oplus g}$ with respect to the right basis, while~$w_1$ vanishes.
On the other hand, if~$\Sigma$ is non-orientable, then there is a unique integer~$h\ge 1$ such that~$\Sigma$ is homeomorphic to the connected sum of~$h$ real projective planes~$\R P^2$,
and~$H_1(\Sigma;\Z_2)$ has dimension~$h$. Note that~$\Sigma$ can also be described as the connected sum of the orientable surface of genus~$\frac{h-1}{2}$ with~$\R P^2$, if~$h$ is odd,
and as the connected sum of the orientable surface of genus~$\frac{h}{2}-1$ with the Klein bottle~$\K$, if~$h$ is even. It follows that, with respect to the right basis,
the intersection form is given by the matrix~$\left(\begin{smallmatrix}0&1\cr 1&0\end{smallmatrix}\right)^{\oplus (h-1)/2}\oplus (1)$ in the first case, and by
the matrix~$\left(\begin{smallmatrix}0&1\cr 1&0\end{smallmatrix}\right)^{\oplus h/2-1}\oplus\left(\begin{smallmatrix}0&1\cr 1&1\end{smallmatrix}\right)$ in the second.
The homomorphism~$w_1$ vanishes on the first~$h-1$ basis elements, and takes value~$1$ on the last one.

As already discussed in the introduction,~$w_1$ defines a homomorphism~$\pi_1(\Sigma)\to\Z_2$, and therefore, a~$2$-fold
cover~$p\colon\widetilde{\Sigma}\to\Sigma$ called the {\em orientation cover\/}. If~$\Sigma$ is orientable, then this cover is trivial.
If~$\Sigma$ is non-orientable on the other hand, then~$\widetilde{\Sigma}$ is a connected orientable closed surface. An easy Euler characteristic computation shows that
if~$\Sigma$ is of genus~$h\ge 1$, then~$\widetilde{\Sigma}$ is of genus~$g=h-1$. For example, the orientation cover of the projective plane~$\R P^2$ is given by the sphere~$S^2$,
while the orientation cover of the Klein bottle~$\K$ is given by the torus~$\mathbb{T}$ (see Figure~\ref{fig:coverK}).

We now turn to quadratic enhancements. Let~$H$ be a~$\Z_2$-vector space endowed with a~$\Z_2$-valued bilinear form~$(\alpha,\beta)\mapsto \alpha\cdot\beta$ and a
linear form~$w\colon H\to\Z_2$. A {\em quadratic enhancement\/}~\cite[Section~3]{K-T} on~$(H,\cdot,w)$ is a map~$q\colon H\to\Z_4$ such that
\[
q(\alpha)-w(\alpha)\;\text{ belongs to }\;2\Z_2\subset\Z_4\quad\text{and}\quad q(\alpha+\beta)=q(\alpha)+q(\beta)+2(\alpha\cdot\beta)\in\Z_4
\]
for all~$\alpha,\beta\in H$. (Here,~$x\mapsto 2x$ denotes the inclusion homomorphism~$2\colon\Z_2\to\Z_4$.)
Note that there are exactly~$|H|$ quadratic enhancements on~$(H,\cdot,w)$, as the set of such forms is an affine space over~$\mathit{Hom}(H;\Z_2)$.
Note also that if~$w$ vanishes, then a quadratic enhancement on~$(H,\cdot,0)$ is equivalent to a map~$q'=\frac{1}{2}q\colon H\to\Z_2$ such
that~$q'(\alpha+\beta)=q'(\alpha)+q'(\beta)+\alpha\cdot\beta$ for all~$\alpha,\beta\in H$; this is known as a {\em quadratic form\/} on~$(H,\cdot)$.

If the triple~$(H,\cdot,w)$ is given by~$H_1(\Sigma;\Z_2)$, the intersection form, and the linear map~$w_1$, then we simply speak of a {\em quadratic enhancement on~$\Sigma$\/}
(and of a {\em quadratic form on~$\Sigma$\/} in the orientable case). By the discussion above, there are exactly~$|H_1(\Sigma;\Z_2)|$ quadratic enhancements
on~$\Sigma$, i.e.~$2^{2g}$ if~$\Sigma$ is orientable of genus~$g$, and~$2^h$ if~$\Sigma$ is non-orientable of genus~$h$.

In order to state our results in the most general form, we will need the following proposition.

\begin{proposition}
\label{prop:q}
Let~$\Sigma$ be a non-orientable surface, and let~$p\colon\widetilde{\Sigma}\to\Sigma$ be its orientation cover. Given any quadratic enhancement~$q$ on~$\Sigma$, there
exists a unique quadratic form~$\widetilde{q}$ on~$\widetilde{\Sigma}$ with the following properties.
\begin{enumerate}[(i)]
\item{If~$C\subset\Sigma$ is a simple closed curve with~$w_1(C)=1$, then~$\widetilde{q}(\widetilde{C})=0$, where~$\widetilde{C}=p^{-1}(C)$.}
\item{If~$C\subset\Sigma$ is a simple closed curve with~$w_1(C)=0$, then~$\widetilde{q}(\widetilde{C}')=\widetilde{q}(\widetilde{C}'')=\frac{1}{2}q(C)$,
where~$\widetilde{C}'\sqcup\widetilde{C}''=p^{-1}(C)$.}
\end{enumerate}
\end{proposition}

\begin{proof}
Let us first assume that the genus~$h$ of~$\Sigma$ is odd. Then,~$\Sigma$ is homeomorphic to the connected sum of an orientable surface~$\Sigma_g$ of genus~$g=\frac{h-1}{2}$
with~$\R P^2$, so its homology splits as a direct sum~$H_1(\Sigma;\Z_2)= H_1(\Sigma_g;\Z_2)\oplus H_1(\R P^2;\Z_2)$. Let us fix a collection of simple closed
curves~$\alpha_1,\dots,\alpha_{h-1}$ representing a basis of~$H_1(\Sigma_g;\Z_2)$, and a non-trivial simple closed curve~$\beta\subset\R P^2$.
The orientation cover~$\widetilde{\Sigma}$ is obtained by gluing two copies of~$\Sigma$ cut along~$\beta$; therefore, a basis of~$H_1(\widetilde{\Sigma};\Z_2)$ is given
by the simple closed curves~$\widetilde\alpha'_1,\widetilde\alpha''_1,\dots,\widetilde\alpha'_{h-1},\widetilde\alpha''_{h-1}$, and condition~$(ii)$ alone determines~$\widetilde{q}$
uniquely. This~$\widetilde{q}$ automatically satisfies condition~$(i)$: indeed, any cycle~$C$ with~$w_1(C)=1$ represents a homology class of the form~$\alpha+\beta$ with~$\alpha\in H_1(\Sigma_g;\Z_2)$, which lifts to~$\widetilde{C}=\widetilde{\alpha}'+\widetilde{\alpha}''+\widetilde{\beta}$. Since the lift~$\widetilde{\beta}$ represents the trivial homology
class,~$\widetilde{q}(\widetilde{C})$ vanishes by condition~$(ii)$.

If~$h$ is even, then~$\Sigma$ is homeomorphic to the connected sum of an orientable surface~$\Sigma_g$ of genus~$g=\frac{h}{2}-1$ with the Klein bottle~$\K$,
so~$H_1(\Sigma;\Z_2)$ splits as~$H_1(\Sigma_g;\Z_2)\oplus H_1(\K;\Z_2)$ and one can fix a collection of simple closed curves~$\alpha_1,\dots,\alpha_{h-2}$ representing a
basis of~$H_1(\Sigma_g;\Z_2)$ together with simple closed curve~$\beta_1,\beta_2$ representing a basis of~$H_1(\K;\Z_2)$ with~$\beta_1\cdot\beta_2=1$,~$w_1(\beta_1)=0$
and~$w_1(\beta_2)=1$. Then,~$\widetilde{\Sigma}$ is obtained by gluing two copies of~$\Sigma$ cut along~$\beta_1$, so a basis of~$H_1(\widetilde{\Sigma};\Z_2)$ is given
by~$\widetilde\alpha'_1,\widetilde\alpha''_1,\dots,\widetilde\alpha'_{h-2},\widetilde\alpha''_{h-2},\widetilde\beta'_1,\widetilde\beta_2$. (Note that~$\widetilde\beta'_1$
and~$\widetilde\beta''_1$ represent the same homology class, see Figure~\ref{fig:coverK}.) Therefore, conditions~$(i)$ and~$(ii)$ uniquely determine~$\widetilde{q}$.
\end{proof}

\begin{figure}[t]
\labellist\small\hair 2.5pt
\pinlabel {$\beta_1$} at 680 175
\pinlabel {$\beta_1$} at 470 175
\pinlabel {$\beta_2$} at 520 225
\pinlabel {$\beta_2$} at 520 -5
\pinlabel {$\widetilde\beta''_1$} at 190 185
\pinlabel {$\widetilde\beta'_1$} at 360 185
\pinlabel {$\widetilde\beta'_1$} at -10 185
\pinlabel {$\widetilde\beta_2$} at 160 -10
\pinlabel {$\widetilde\beta_2$} at 160 225
\pinlabel {$\to$} at 415 110
\pinlabel {$\mathbb{T}$} at -30 115
\pinlabel {$\K$} at 700 115
\endlabellist
\centering
\includegraphics[width=0.6\textwidth]{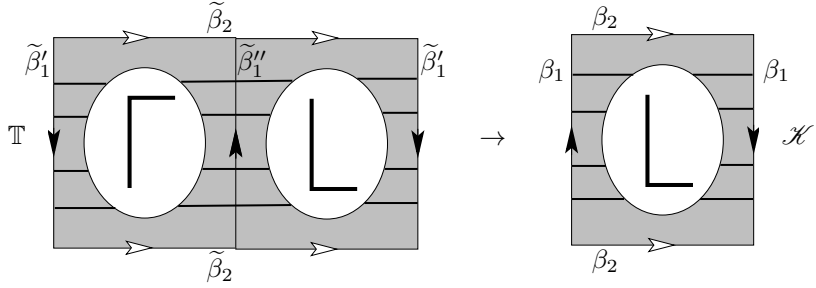}
\caption{A schematic description of the orientation cover of the Klein bottle by the torus, with explicit bases of~$H_1(\K;\Z_2)$ and~$H_1(\mathbb{T};\Z_2)$.}
\label{fig:coverK}
\end{figure}

\subsection{The general identities}
\label{sub:general}

Let~$G$ be a weighted graph embedded in a surface~$\Sigma$. Given any~$\alpha\in H_1(\Sigma;\Z_2)$ and~$D_0\in\D(G)$, let us define
\[
Z_{\alpha,D_0}(G)=\sum_{[D\Delta D_0]=\alpha}\nu(D)\,,
\]
where the sum is over all~$D\in\D(G)$ such that the homology class of the symmetric difference~$D\Delta D_0$ is equal to~$\alpha$.

\begin{theorem}
\label{thm:main}
Let~$G$ be a weighted graph with an even number of vertices, embedded in a non-orientable surface~$\Sigma$ in such a way that its complement consists of topological discs,
and let~$\widetilde{G}\subset\widetilde{\Sigma}$ denote the weighted graph obtained by lifting~$G$ via the orientation cover~$\widetilde{\Sigma}\to\Sigma$.
Then, for any quadratic enhancement~$q$ on~$\Sigma$ and any~$D_0\in\D(G)$,~$\widetilde{D}_0\in\D(\widetilde{G})$, we have
\[
\Big|\sum_{\widetilde\alpha}(-1)^{\widetilde{q}(\widetilde\alpha)} Z_{\widetilde\alpha,\widetilde{D}_0}(\widetilde{G})\Big|=
\Big(\sum_{w_1(\alpha)=0}(-1)^{\frac{q(\alpha)}{2}} Z_{\alpha,D_0}(G)\Big)^2+
\Big(\sum_{w_1(\alpha)=1}(-1)^{\frac{q(\alpha)-1}{2}} Z_{\alpha,D_0}(G)\Big)^2\,,
\]
where the first sum is over all~$\widetilde\alpha\in H_1(\widetilde{\Sigma};\Z_2)$, the second (resp. third) over all~$\alpha\in H_1(\Sigma;\Z_2)$ with~$w_1(\alpha)=0$
(resp.~$w_1(\alpha)=1$), and~$\widetilde{q}$ is the quadratic form on~$\widetilde{\Sigma}$ determined by~$q$ as in Proposition~\ref{prop:q}.
Finally, neither side of this equality depends on the choice of~$D_0$ and~$\widetilde{D}_0$.
\end{theorem}

\begin{remark}
\label{rem:count}
Recall that if~$\Sigma$ is of (non-orientable) genus~$h\ge 1$, then there are~$2^h$ quadratic enhancements on it, and~$\widetilde{\Sigma}$ is of (orientable) genus~$g=h-1$.
Furthermore, one easily checks that the map~$q\mapsto\widetilde q$ is two-to-one: indeed, picking a basis of~$H_1(\Sigma;\Z_2)$ with~$w_1$ non-vanising on a single
element,~$\widetilde q$ is not affected by changing the value of~$q$ on this element. As the right-hand side of the equality in Theorem~\ref{thm:main} is also left
unchanged by this transformation, we actually have exactly~$2^{h-1}=2^g$ distinct equalities.
\end{remark}

\begin{example}
\label{ex:h=1}
Let us consider the case where~$\Sigma$ is of genus~$h=1$, i.e. the projective plane~$\R P^2$, whose orientation cover~$\widetilde{\Sigma}$ is the sphere.
There are two quadratic enhancements on~$\R P^2$, assigning the value~$q(\alpha)=1$ or~$q(\alpha)=3$ to the unique non-trivial class~$\alpha=1$ (and the value~$q(\alpha)=0$
to the trivial class~$\alpha=0$), and a single (trivial) quadratic form on the sphere. Therefore, in this case, Theorem~\ref{thm:main} gives a single equality, namely
\[
Z(\widetilde{G})=Z_{0,D_0}(G)^2+Z_{1,D_0}(G)^2\,.
\]
Removing a small disk from~$\Sigma$ and its lifts from~$\widetilde{\Sigma}$, we obtain the~$2$-fold cover~$\C\to\M$ considered in Section~\ref{sec:Mob}.
Hence, the equality above is exactly Proposition~\ref{prop:Mob}, of which an elementary proof was given in Section~\ref{sec:Mob}. 
\end{example}

\begin{example}
\label{ex:h=2}
Let us now consider the case of genus~$h=2$, i.e. the Klein bottle~$\K$ whose orientation cover is the torus~$\mathbb{T}$.
As usual, fix a basis of~$H_1(\K;\Z_2)$ given by two simple closed curves~$\beta_1,\beta_2$ with~$\beta_1\cdot\beta_2=1$,~$w_1(\beta_1)=0$ and~$w_1(\beta_2)=1$, which determines
a basis~$\widetilde\beta'_1,\widetilde\beta_2$ of~$H_1(\mathbb{T};\Z_2)$, as described in Figure~\ref{fig:coverK}.
There are~$4$ quadratic enhancements on~$\K$, depending on the values of~$q(\beta_1)\in\{0,2\}$ and of~$q(\beta_2)\in\{1,3\}$.
By definition, the corresponding quadratic form~$\widetilde{q}$ on~$\mathbb{T}$ vanishes on~$\widetilde\beta_2$ and takes the value~$\frac{q(\beta_1)}{2}$ on~$\widetilde\beta'_1$.
Writing~$\e$ for~$(-1)^{\frac{q(\beta_1)}{2}}$, using the equalities~$q(\beta_1+\beta_2)=q(\beta_1)+q(\beta_2)+2$
and~$\widetilde{q}(\widetilde\beta'_1+\widetilde\beta_2)=\widetilde{q}(\widetilde\beta'_1)+1$, and leaving~$D_0,\widetilde D_0$ implicit, we obtain the two identities
\begin{equation}
\label{equ:h=2}
|Z_{00}(\widetilde{G})+Z_{01}(\widetilde{G})+\e Z_{10}(\widetilde{G})-\e Z_{11}(\widetilde{G})|=\left(Z_{00}(G)+\e Z_{10}(G)\right)^2+\left(Z_{01}(G)-\e Z_{11}(G)\right)^2
\end{equation}
for~$\e=\pm 1$. The proof of Theorem~\ref{thm:2} is now straightforward, as explained below.
\end{example}

\begin{proof}[Proof of Theorem~\ref{thm:2}]
Since~$G\subset\K$ is locally bipartite, the map assigning to each cycle its length induces a well-defined linear map~$\ell\colon H_1(\K;\Z_2)\to\Z_2$, which
is not identically zero since~$G$ is not bipartite. Note that~$Z_\alpha(G)$ vanishes as soon as~$\ell(\alpha)=1$.
By assumption, there exists~$D_0\in\D(G)$ such that the homology class of~$D_0\Delta\tau(D_0)$ is~$\beta_2$ or~$\beta_1+\beta_2$.
Let us first assume that~$[D_0\Delta\tau(D_0)]=\beta_2$.
Since~$D_0\Delta\tau(D_0)$ is of even length, we then have~$\ell(\beta_2)=0$ and therefore~$\ell(\beta_1)=1$ and~$\ell(\beta_1+\beta_2)=1$.
This implies that~$Z_{10}(G)$ vanishes, as well as~$Z_{11}(G)$.
Using the length function on the torus (which is determined by~$\ell$), we find that~$Z_{10}(\widetilde{G})=Z_{11}(\widetilde{G})=0$. Hence, Equation~(\ref{equ:h=2}) boils down to~$Z(\widetilde{G})=Z_{00}(G)^2+Z_{01}(G)^2$.
The facts that~$G$ is invariant by a horizontal tranlation~$\tau$ and that~$[D_0\Delta\tau(D_0)]=\beta_2$ allows us to construct a weight-preserving bijection between the dimer
configurations contributing to~$Z_{00}(G)$ and those contributing to~$Z_{01}(G)$, as in the proof of Theorem~\ref{thm:1}. The result follows.
The case where~$[D_0\Delta\tau(D_0)]=\beta_1+\beta_2$ is similar, and therefore left to the reader.
\end{proof}

\begin{remark}
\label{rem:other}
Obviously, other formulas relating~$Z(\widetilde{G})$ to~$Z(G)$ can be derived from Theorem~\ref{thm:main}.
For example, using Equation~(\ref{equ:h=2}) and the notation in the discussion of this~$h=2$ case, one can show the following statement:
{\em If~$G\subset\K$ is locally bipartite and such that~$\ell(\beta_1)=0$,~$\ell(\beta_2)=1$ and~$Z_{10}(\widetilde{G})=0$ or~$Z_{11}(\widetilde{G})=0$, then~$Z(\widetilde{G})=Z(G)^2$.}
However, such a statement is not as satisfactory as Theorems~\ref{thm:1} and~\ref{thm:2} above, as it does assume a vanishing condition that is difficult to obtain as a consequence
of a natural property of the graph. Note also that the number of such conditions will grow exponentially with the genus~$h$ of~$\Sigma$. For this reason, we do not expect to
obtain very compelling corollaries in higher genus. Of course, Theorem~\ref{thm:main} might nevertheless explain other ``curious identities'' obtained for specific examples.
\end{remark}

We conclude this article with the proof of our main result.

\begin{proof}[Proof of Theorem~\ref{thm:main}]
Fix a weighted graph~$G$ with an even number of vertices, embedded in a non-orientable surface~$\Sigma$ so that~$\Sigma\setminus G$ consists of topological discs.
As explained in subsection~\ref{sub:hom}, one can find a simple closed curve~$\beta$, transverse to the graph~$G$, so that the orientation cover~$\widetilde{\Sigma}$ of~$\Sigma$
can be constructed by gluing along their boundary two copies of~$\Sigma$ cut along~$\beta$. This determines a map~$\omega$ on the edges of~$G$ by setting~$\omega(e)=1$
if~$e$ intersects~$\beta$, and~$\omega(e)=0$ else. With such a data in hand, any orientation~$K$ on the edges of~$G$ induces a modified weighted skew-adjacency matrix~$A^{K,\omega}(G)$
as defined by Equation~(\ref{equ:A'}).

The right notion of a Kasteleyn orientation in this setting in the following generalization of the one introduced in subsection~\ref{sub:Mob}. Consider the weighted
graph~$\widetilde{G}\subset\widetilde{\Sigma}$ obtained taking two copies of~$G\subset\Sigma$ cut along the edges with~$\omega(e)=1$, and glued back again along these edges, each such
edge joining the two copies. Given an orientation~$K$ on~$G$, let~$\widetilde{K}$ be the orientation on~$\widetilde{G}$ obtained by lifting~$K$ to~$\widetilde{G}$, and inverting it on
all edges of~$\widetilde{G}$ completely contained in one of the copies of~$\Sigma$ cut along~$\beta$. Then, we shall say that~$K$ is a Kasteleyn orientation
on~$(G\subset\Sigma,\omega)$ if~$\widetilde{K}$ is a Kasteleyn orientation on~$\widetilde{G}\subset\widetilde{\Sigma}$, i.e: for any face~$\widetilde{f}$
of~$\widetilde{G}\subset\widetilde{\Sigma}$,~$n^{\widetilde{K}}(\partial\widetilde{f})$ is odd, where the boundary of~$\widetilde{f}$ is oriented according to a fixed orientation
on the surface~$\widetilde{\Sigma}$.

By~\cite[Theorem~4.3]{Cim},~$(G\subset\Sigma,\omega)$ admits a Kasteleyn orientation if and only if~$G$ has an even number of vertices, which we assumed. As proved
in~\cite[p.174]{Cim}, if~$K$ is such an orientation, then for any~$D_0\in\D(G)$, we have
\[
i^{-\omega(D_0)}\,\e^K(D_0)\,\Pf(A^{K,\omega}(G))=\sum_{\alpha\in H_1(\Sigma;\Z_2)}i^{-q(\alpha)}Z_{\alpha,D_0}(G)\,,
\]
where~$q=q^{K,\omega}_{D_0}$ is a quadratic enhancement on~$\Sigma$. By~\cite[Theorem~5.3]{Cim}, any quadratic enhancement on~$\Sigma$ can be obtained in such a way by varying the
orientation~$K$. Applying the~$\omega=0$ case of this formula to~$\widetilde{G}\subset\widetilde{\Sigma}$, we get
\[
\e^{\widetilde{K}}(\widetilde{D}_0)\,\Pf(A^{\widetilde{K}}(\widetilde{G}))=
\sum_{\widetilde\alpha\in H_1(\widetilde\Sigma;\Z_2)}(-1)^{\widetilde q(\widetilde\alpha)}Z_{\widetilde\alpha,\widetilde{D}_0}(\widetilde{G})\,,
\]
where~$\widetilde q=\widetilde{q}^{\widetilde{K}}_{\widetilde{D}_0}$ is a quadratic form on~$\widetilde\Sigma$.
The proof of Lemma~\ref{lemma:Pf} can be applied {\em verbatim\/}, leading to~$|\Pf(A^{\widetilde{K}}(\widetilde{G}))|=|\Pf(A^{K,\omega}(G))|^2$. This equality,
together with the two equations displayed above, and the fact that a quadratic enhancement always satisfies~$q(\alpha)-w_1(\alpha)\in 2\Z_2$, lead to
\[
\Big|\sum_{\widetilde\alpha}(-1)^{\widetilde{q}(\widetilde\alpha)} Z_{\widetilde\alpha,\widetilde{D}_0}(\widetilde{G})\Big|=
\Big(\sum_{w_1(\alpha)=0}(-1)^{\frac{q(\alpha)}{2}} Z_{\alpha,D_0}(G)\Big)^2+
\Big(\sum_{w_1(\alpha)=1}(-1)^{\frac{q(\alpha)-1}{2}} Z_{\alpha,D_0}(G)\Big)^2\,.
\]
Also, we see by the same two equations displayed above that neither side of this equality depends on the choice of~$D_0$ and~$\widetilde{D}_0$.

Therefore, we are left with the proof that the map~$q=q^{K,\omega}_{D_0}\mapsto\widetilde{q}^{\widetilde{K}}_{\widetilde{D}_0}=\widetilde{q}$ corresponds to the
map~$q\mapsto\widetilde{q}$ described in Proposition~\ref{prop:q}. To do so, let us assume without loss of generality that~$\widetilde{D}_0$ is the lift of~$D_0$.
By~\cite[Proposition~3.7]{Cim},~$\widetilde{q}$  satisfies the following property: if~$\widetilde{C}$ is an oriented simple closed curve
in~$\widetilde{G}\subset\widetilde{\Sigma}$, then
\[
\widetilde{q}(\widetilde{C})=n^{\widetilde{K}}(\widetilde{C})+\ell_{\widetilde{D}_0}(\widetilde{C})+1\,,
\]
where the second term denotes the number of vertices of~$\widetilde{C}$ whose adjacent dimer of~$\widetilde{D}_0$ sticks out to the left of~$\widetilde{C}$. (Recall that both
~$\widetilde{C}$ and~$\widetilde{\Sigma}$ are oriented.) To check the first condition of Proposition~\ref{prop:q}, let us apply this formula to the lift of an oriented simple
closed curve~$C$ in~$G\subset\Sigma$ with~$w_1(C)=1$. Writing~$C^0$ (resp.~$C^1$) for the edges of~$C$ with~$\omega(e)=0$ (resp.~$\omega(e)=1$), and computing modulo~$2$, we get
\[
n^{\widetilde{K}}(\widetilde{C})=n^K(C^0)+n^K(-C^0)+2n^K(C^1)=|C^0|=|C|+w_1(C)=|C|+1\,.
\]
Furthermore,~$\ell_{\widetilde{D}_0}(\widetilde{C})$ is given by the number of vertices of~$C$ whose adjacent dimer of~$D_0$ sticks out of~$C$. Since this number has
the same parity as~$|C|$, we have~$\widetilde{q}(\widetilde{C})=0$ as expected. To show that~$\widetilde{q}$ satisfies the second condition
of Proposition~\ref{prop:q}, first note that we can assume~$C$ to be contained in one copy of~$\Sigma$ cut along~$\beta$. In such a case,~\cite[Proposition~5,2]{Cim} gives the equality
\[
q(C)=2(n^{K}(C)+\ell_{D_0}(C)+1)\,,
\]
where the vertices contributing to~$\ell_{D_0}(C)$ are counted using an orientation on one of the copies of~$\Sigma$ cut along~$\beta$. By the value
of~$\widetilde{q}(\widetilde{C}')$ displayed above, the second condition is verified, and the proof complete.
\end{proof}

\bibliographystyle{plain}
\nocite{*}
\bibliography{bibliographie}

\end{document}